\documentclass[twocolumn,english]{revtex4-1}
\usepackage[T1]{fontenc}
\usepackage[utf8]{inputenc}
\usepackage{babel}
\usepackage{enumitem}
\usepackage{amsmath}
\usepackage{amsthm}
\usepackage{amssymb}
\usepackage{graphicx}
\usepackage[unicode=true,pdfusetitle,
 bookmarks=true,bookmarksnumbered=false,bookmarksopen=false,
 breaklinks=false,pdfborder={0 0 1},backref=false,colorlinks=false]
 {hyperref}

\makeatletter
  \theoremstyle{definition}
  \newtheorem{defn}{\protect\definitionname}
  \theoremstyle{plain}
  \newtheorem*{thm*}{\protect\theoremname}
  \theoremstyle{plain}
  \newtheorem{fact}{\protect\factname}
  \theoremstyle{plain}
  \newtheorem{cor}{\protect\corollaryname}
  \theoremstyle{remark}
  \newtheorem{claim}{\protect\claimname}

\usepackage{bbm}

\@ifundefined{showcaptionsetup}{}{%
 \PassOptionsToPackage{caption=false}{subfig}}
\usepackage{subfig}
\makeatother

  \providecommand{\claimname}{Claim}
  \providecommand{\definitionname}{Definition}
  \providecommand{\factname}{Fact}
  \providecommand{\theoremname}{Theorem}
\providecommand{\corollaryname}{Corollary}

\begin{document}

\title{On the local equivalence of complete bipartite and repeater graph
states}

\author{Ilan Tzitrin}
\email{itzitrin@physics.utoronto.ca}

\selectlanguage{english}%

\affiliation{Department of Physics, University of Toronto, Toronto, Ontario, M5S
1A7}
\begin{abstract}
Classifying locally equivalent graph states, and stabilizer states
more broadly, is a significant problem in the theories of quantum
information and multipartite entanglement. A special focus is given
to those graph states for which equivalence through local unitaries
implies equivalence through local Clifford unitaries (LU $\Leftrightarrow$
LC). Identification of locally equivalent states in this class is
facilitated by a convenient transformation rule on the underlying
graphs and an efficient algorithm. Here we investigate the question
of local equivalence of the graph states behind the all-photonic quantum
repeater. We show that complete bipartite graph (biclique) states,\emph{
imperfect} repeater graph states and small ``crazy graph'' states
satisfy LU $\Leftrightarrow$ LC. We continue by discussing biclique
states more generally and placing them in the context of counterexamples
to the LU-LC Conjecture. To this end, we offer some alternative proofs
and clarifications on existing results.
\end{abstract}
\maketitle

\section{Introduction}

Graph states are a subset of multipartite entangled states used pervasively
in quantum information protocols: quantum error-correcting codes,
measurement-based quantum computing, entanglement purification, information
splitting, quantum cryptography, and all-photonic quantum repeaters,
among others \citep{Briegel2001,Hein2004,Hein2006a,Varnava2006,Browne2005,Nielsen2006,Aschauer2005b,Aschauer2005,Kruszynska2006,Muralidharan2008,Keet2010,Chen2004,Dur2005,Azuma2015}.
To define a graph state one starts by specifying the underlying (undirected)
graph, $G$, through a set of vertices, $V$, and a collection of
edges, $E$. Vertices are then associated with qubits and edges with
entangling gates:
\begin{equation}
\left|G\right\rangle \equiv\prod_{\left\{ a,b\right\} \in E}\text{CZ}^{ab}\left|+\right\rangle ^{V},\label{eq:gdef}
\end{equation}
where $\text{CZ}^{ab}$ is the controlled phase gate with control
qubit $a$ and target qubit $b$, $\left|+\right\rangle \equiv\frac{1}{\sqrt{2}}\left(\left|0\right\rangle +\left|1\right\rangle \right)$,
and we use the notation $\left|\psi\right\rangle ^{V}\equiv\bigotimes_{a\in V}\left|\psi\right\rangle ^{a}$
\citep{Hein2006a}. The value of graph states lies in the ability
of local measurements to transform them in specified ways. In fact,
any two qubits in a connected graph state can be projected to a Bell
pair through a sequence of Pauli $X$ and $Z$ measurements, in what
is termed\emph{ localizable entanglement }\citep{Hein2006a}\emph{.}

A study of the properties of graph states is an active area of research
made difficult by their size and complexity. As the number of qubits
increases, determining equivalence classes of graph states under local
unitaries (LUs) or invertible stochastic local operations and classical
communication (SLOCC) quickly becomes intractable \citep{Zeng2007a,Verstraete2002}.

On the other hand, because graph states are tied to the stabilizer
formalism, as defined below, there are advantages in looking at a
smaller class of operations, the \emph{local Clifford unitaries }(LCs),
which map the group of Pauli operators to itself under conjugation.
For one, an LC operation effects a simple graphical transformation,
the \emph{local complementation}, on the underlying graph, as shown
in Fig. \ref{fig:localcomp} \citep{Nest2003}. Furthermore, there
is a classical polynomial-time algorithm that decides whether two
given graphs are related by a sequence of local complementations \citep{Bouchet1991,VanNest2004}.

Initially, it was conjectured that two graph states equivalent under
local unitaries are also equivalent under local Clifford gates, that
is, LU $\Leftrightarrow$ LC \citep{Nest2003,Krueger2005}. One direction
of the conjecture - that LC equivalence implies LU equivalence - is
trivial, since local Clifford gates are local unitaries themselves.
The other direction, however, is false, with the smallest counterexample
known at 27 qubits \citep{Ji2007,Guhne2017}. 

In this paper, we seek to modestly expand the class of states for
which the status of LU $\Leftrightarrow$ LC is known. After overviewing
the stabilizer formalism and some important existing results on the
local equivalence of graph states in Section \ref{sec:Preliminaries},
we use Section \ref{subsec:CBR} to present our findings on the repeater
graph states displayed in Fig. \ref{fig:rgs}. These are the photonic
states used in \citep{Azuma2015} to eschew quantum memories in traditional
matter-based quantum repeaters. Section \ref{subsec:Bicliques} is
reserved for a broader discussion of biclique states (see Fig. \ref{fig:biclique}),
which appear in protocols beyond all-photonic repeaters \citep{Chen2004,Aschauer2005,Dur2005,Rudolph2016b}.
Finally, in, Section \ref{subsec:counterex}, we clarify the result
on counterexamples to the LU-LC Conjecture by Ji et al. \citep{Ji2007}
and place it in the context of bipartite graphs.

\begin{figure}
\begin{centering}
\includegraphics[scale=0.25]{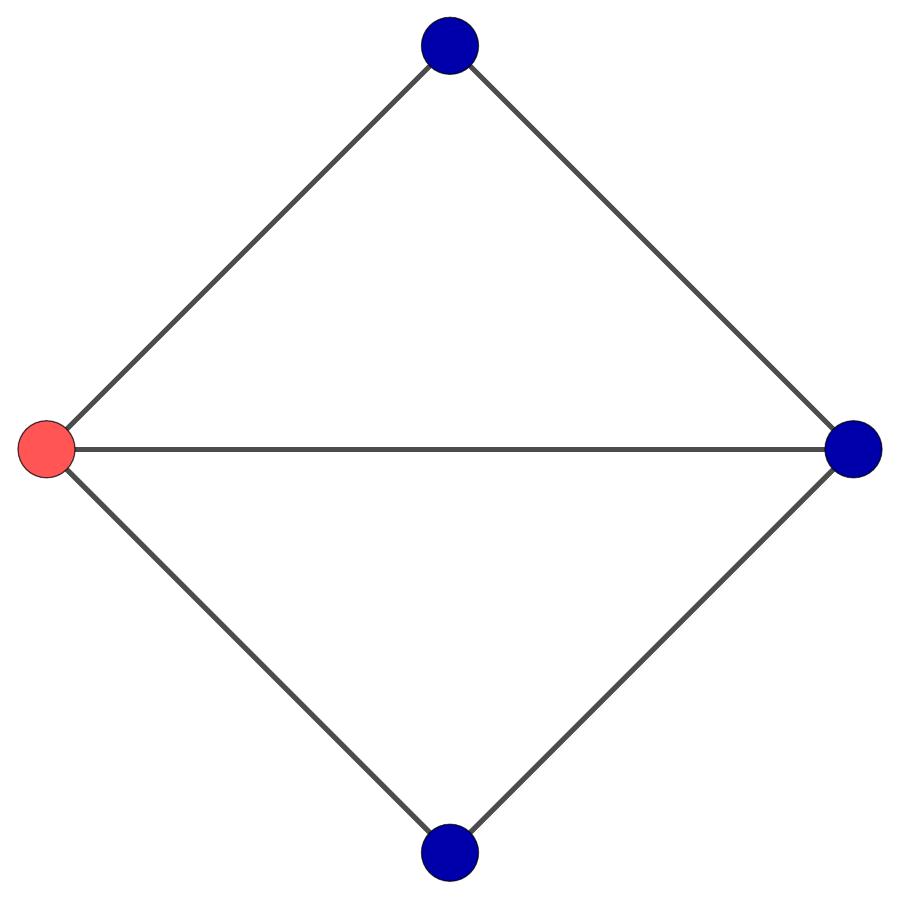}$\overset{\text{LC}}{\longleftrightarrow}$\includegraphics[scale=0.25]{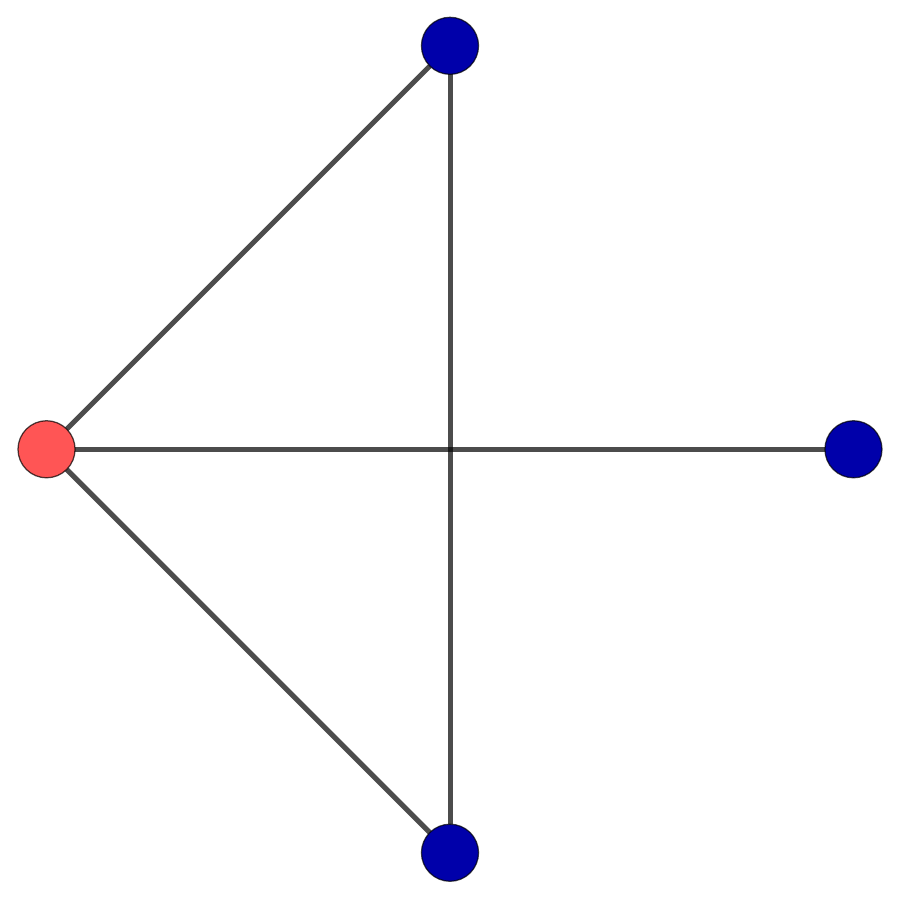}
\par\end{centering}
\caption{Applying a local complementation to the red (lighter) vertex in either
graph produces the other graph. \label{fig:localcomp} The effect
of a local complementation on a vertex is to remove the existing edges
among its neighbours and add the missing ones. In light of the LC
Rule in Section \ref{subsec:Review-of-results}, the letters LC in
the figure can refer both to a local complementation and to an appropriate
local Clifford gate.}
\end{figure}
\begin{figure}
\begin{centering}
\subfloat[]{\centering{}\includegraphics[scale=0.45]{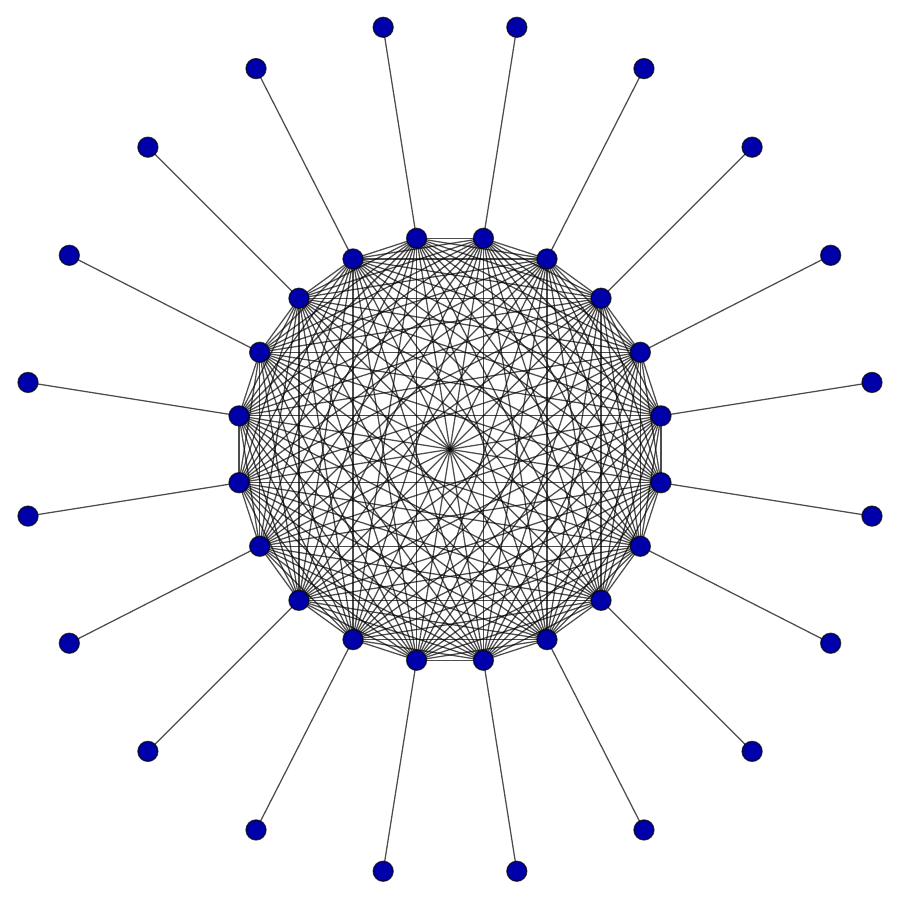}}\subfloat[]{\begin{centering}
\includegraphics[scale=0.45]{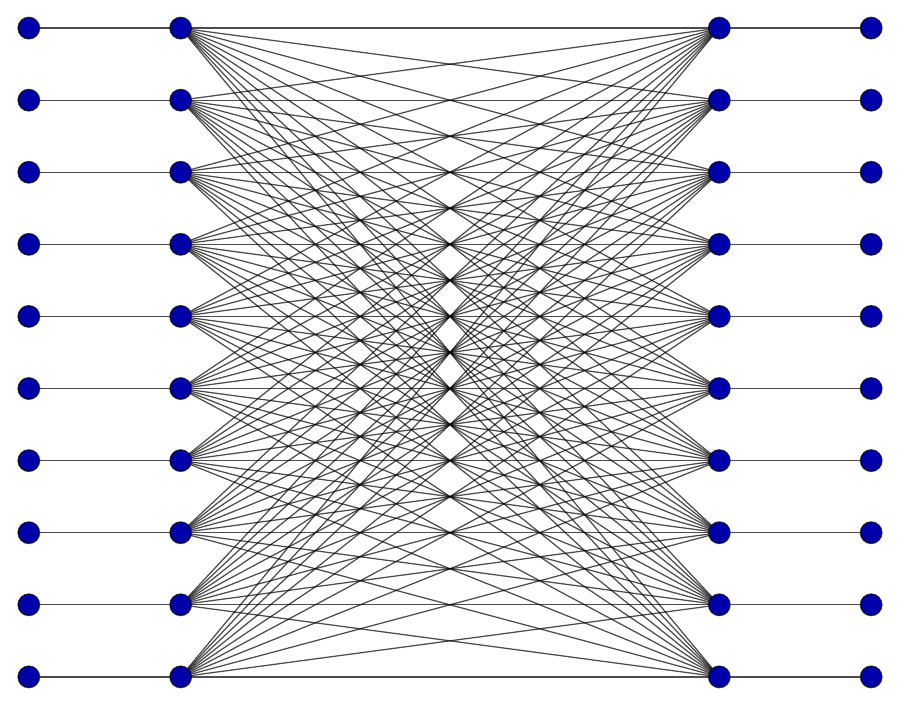}
\par\end{centering}
}
\par\end{centering}
\caption{20-vertex repeater graphs: (a) a complete graph and (b) a biclique
(or complete bipartite graph) with leaves appended to each vertex.
In our context, the graphs underlie all-photonic quantum repeater
graph states $\left|R_{C}^{20}\right\rangle $ (from \citep{Azuma2015})
and $\left|R_{B}^{20}\right\rangle $, respectively. The vertices
correspond to qubits, while the edges represent entanglement localizable
through sequences of single-qubit measurements. \label{fig:rgs}}
\end{figure}

\section{Preliminaries \label{sec:Preliminaries}}

In this section, we provide the necessary preliminaries to understand
our results. We describe the stabilizer formalism and then enumerate
some known results on the LU-LC question, which we use as a foundation
for the conclusions in Section \ref{sec:Results}.

\subsection{Stabilizer formalism and Minimal Support Condition}

The \emph{n-qubit Pauli group}, $\mathcal{P}_{n}$, consists of operators
of the form $P_{1}\otimes\cdots\otimes P_{n}$, where the $P_{i}$s
are either the Pauli matrices or the $2\times2$ identity matrix,
and we allow a phase of $\pm1$ or $\pm i$. A \emph{stabilizer }is
a commutative subgroup of $\mathcal{P}_{n}$ so that $-I$ is not
an element of the group. A \emph{stabilizer state }is the unique eigenstate
(with eigenvalue 1) of its stabilizer \citep{Nielsen2010}. A graph
state, $\left|G\right\rangle $, can be defined equivalently as the
state stabilized by the set of observables $\left\{ K_{a}\vert a\in V_{G}\right\} $,
where
\begin{equation}
K_{a}\equiv X^{a}\bigotimes_{b\in N_{a}}Z^{N_{a}},
\end{equation}

and $N_{a}$ denotes the\emph{ neighbourhood} of $a$: the set of
vertices connected (adjacent) to $a$, not including $a$ itself \citep{Hein2006a}.
The stabilizer element, $K_{\ell}$, corresponding to a \emph{leaf},
$\ell$ (a vertex attached to only a single parent, $p$, as in Fig.
\ref{fig:rgs}) is therefore $X^{\ell}Z^{p}$. Stabilizer generators
can be conveniently represented as linearly independent rows of an
$n\times2n$ binary \emph{check matrix }$S=\left[X\vert Z\right]$
where the left (right) sides have 1s to indicate the indices of $X$
$\left(Z\right)$ operators, and 0 otherwise. The check matrix of
a graph state, $\left|G\right\rangle ,$ is then $S_{G}=\left[I\vert\Gamma\right]$,
where $\Gamma$ is the graph's \emph{adjacency matrix}, that is 
\begin{equation}
\Gamma_{ij}=\begin{cases}
0 & \left\{ i,j\right\} \notin E_{G}\\
1 & \left\{ i,j\right\} \in E_{G}
\end{cases}.
\end{equation}

In the discussion that follows, we will refer to the Minimal Support
Condition (MSC) \citep{Nest2008}, a statement about the structure
of a stabilizer. We therefore list some definitions important for
understanding the condition, using a leaf qubit, $\ell$, for illustration:
\begin{itemize}
\item The \emph{weight} of a stabilizer element is the number of non-identity
operators comprising it. The weight of $K_{\ell}$ is 2.
\item The\emph{ distance }of a stabilizer state is the weight of the lowest-weight
element in its stabilizer. Any connected graph state with leaves is
therefore a distance-2 state.
\item The \emph{support, $\omega_{a}$, }of a stabilizer element is the
set of indices of its non-identity elements. For example, $\omega_{\ell}=\left\{ \ell,p\right\} $.
The support is called \emph{minimal }if there is no other stabilizer
element $K_{b}$ with support $\omega_{b}$ so that $\omega_{b}\subset\omega_{a}$.
A \emph{minimal element} is a stabilizer element with minimal support.
For any connected graph state with leaves, then, $\omega_{\ell}$
is automatically a minimal support, and $K_{\ell}$ a minimal element. 
\end{itemize}
\begin{defn}[Minimal Support Condition \citep{Nest2008}]
 Let $\left|G\right\rangle $ be a fully connected graph state with
stabilizer $\left\langle K_{a}\right\rangle $. Let $\mathcal{\mathcal{M}}$
denote the subgroup of $\left\langle K_{a}\right\rangle $ generated
by its minimal elements. Then $\left|G\right\rangle $ satisfies the
\emph{Minimal Support Condition (MSC) }iff the Pauli operators $X$,
$Y$ and $Z$ each occur on every qubit in $\mathcal{M}$. \label{def:MSC}
\end{defn}
States satisfying the MSC have a rich stabilizer structure with repercussions
for their $\text{LU}\Leftrightarrow\text{LC}$ status, as shown in
the following section (see also \citep{Nest2008}).

\subsection{Review of results on locally equivalent graph states \label{subsec:Review-of-results}}

Let $\left|G\right\rangle $ and $\left|\tilde{G}\right\rangle $
be two $n$-qubit graph states. If there exists a local unitary, $U=\bigotimes_{i=1}^{n}U_{i}$,
so that $U\left|G\right\rangle =\left|\tilde{G}\right\rangle $, then
the graph states are said to be \emph{LU-equivalent.} If there exists
a local Clifford unitary, $U_{C}=\bigotimes_{i=1}^{n}U_{Ci}$, so
that $U_{C}\left|G\right\rangle =\left|\tilde{G}\right\rangle $,
the graph states are \emph{LC-equivalent}. Local Clifford unitaries
are defined through their effect on the Pauli group: $A\in\mathcal{P}_{n}\implies U_{C}AU_{C}^{\dagger}\in\mathcal{P}_{n}$
\citep{Nielsen2010}.

One important repercussion of the LC-equivalence of two graph states
is that their underlying graphs can be transformed between each other
with a simple graphical prescription known as the LC Rule, expressed
in the following theorem:
\begin{thm*}[LC Rule \citep{Hein2006a,Nest2003}]
\emph{ }Two graph states, $\left|G\right\rangle $ and $\left|\tilde{G}\right\rangle $,
are LC-equivalent iff the graphs $G$ and $\tilde{G}$ are related
through a sequence of local complementations. The LC unitary that
effects a local complementation about vertex $a$ is given by
\begin{equation}
U_{a}^{\text{LC}}\left(G\right)=e^{-i\frac{\pi}{4}X^{a}}e^{i\frac{\pi}{4}Z^{N_{a}}}\propto\sqrt{K_{a}}.
\end{equation}
\end{thm*}
\begin{proof}
Given in \citep{Hein2006a,Nest2003}.
\end{proof}
While there exist graph states for which local Clifford equivalence
is not sufficient to describe local unitary equivalence more broadly
\citep{Ji2007,Guhne2017}, $\text{LU}\Leftrightarrow\text{LC}$ holds
for a large assortment of graph states. We can immediately spot some
of these states with a simple fact:
\begin{fact}
If $\text{LU}\Leftrightarrow\text{LC}$ holds for a graph state, $\left|G\right\rangle $,
then it holds for any graph state, $\left|\tilde{G}\right\rangle $,
LC-equivalent to $\left|G\right\rangle $. \label{fact:LCclass}
\end{fact}
\begin{proof}
Suppose that $\left|G\right\rangle $ satisfies $\text{LU}\Leftrightarrow\text{LC}$.
Since $\left|G\right\rangle $ and $\left|\tilde{G}\right\rangle $
are LC-equivalent, we may find a local Clifford unitary, $U_{C}$,
so that $\left|G\right\rangle =U_{C}\left|\tilde{G}\right\rangle $.
Let $U$ be an arbitrary local unitary, and suppose there is a graph,
$H$, so that $U\left|\tilde{G}\right\rangle =\left|H\right\rangle $.
Then
\begin{align*}
U\left|\tilde{G}\right\rangle =\left|H\right\rangle  & \implies U\left(U_{C}^{-1}\left|G\right\rangle \right)=\left|H\right\rangle \\
 & \implies\left(UU_{C}^{-1}\right)\left|G\right\rangle =\left|H\right\rangle .
\end{align*}

As $\text{LU}\Leftrightarrow\text{LC}$ holds for $\left|G\right\rangle $
and $UU_{C}^{-1}$ is a local unitary, it is true that there exists
a local Clifford unitary, $V_{C}$, with
\[
V_{C}\left|G\right\rangle =\left|H\right\rangle \implies V_{C}U_{C}\left|\tilde{G}\right\rangle =\left|H\right\rangle .
\]

Since the product of two local Clifford unitaries is again a local
Clifford unitary, the result follows.
\end{proof}
Because any stabilizer state is LC-equivalent to a graph state, Fact
\ref{fact:LCclass} implies that answering a question related to the
local equivalence of graph states has immediate repercussions for
stabilizer states \citep{VanNest2004,Schlingemann2001,Grassl2002}. 

The following is a sample of six significant Results in the literature
that we refer to in Section \ref{sec:Results}. It is known that $\text{LU}\Leftrightarrow\text{LC}$
holds for $\left|G\right\rangle $ if
\begin{enumerate}
\item $G$ has 8 vertices or fewer \citep{Hein2006a,Cabello2009}. \label{enu:-has-8}
\item $\left|G\right\rangle $ is the $n$-Greenberger-Horne-Zeilinger (GHZ)
state: $\left|\text{GHZ}_{n}\right\rangle \equiv\frac{1}{\sqrt{2}}\left(\left|0\right\rangle ^{\otimes n}+\left|1\right\rangle ^{\otimes n}\right)$
\citep{Greenberger1989,Nest2008}. \label{enu:nghz}
\item $\left|G\right\rangle $ satisfies the Minimal Support Condition given
in Def. (\ref{def:MSC}) \citep{Nest2008}. \label{enu:msc}
\item The state obtained by removing the leaves from $G$ satisfies the
MSC \citep{Zeng2007a}. \label{enu:mscleaf}
\item $G$ has no cycles of length 3 or 4. (Here a \emph{cycle }denotes
any path that starts and ends at the same vertex, with no edges repeated,
and the \emph{length} is the number of edges along this path.) \citep{Zeng2007a}.
\label{enu:cycle}
\item The stabilizer of $\left|G\right\rangle $ has rank of support less
than 6 \citep{Ji2007}. \label{enu:rankstab}
\end{enumerate}
Armed with the tools from this section, we are ready to investigate
LU-LC equivalence in all-photonic quantum repeater graph states.

\section{Results \label{sec:Results}}

\subsection{Complete graph, biclique and repeater graph states \label{subsec:CBR}}

We begin this section by showing that $n$-GHZ states satisfy $\text{LU}\Leftrightarrow\text{LC}$.
Although this result was shown in \citep{Nest2008}, it is instructive
to present it as a simple corollary to the results in Section \ref{sec:Preliminaries}
to contrast the direct approach taken by Van den Nest et al.:
\begin{cor}
\label{cor:n-GHZ}n-GHZ states satisfy $\text{LU}\Leftrightarrow\text{LC}$.
\end{cor}
\begin{proof}
$\left|\text{GHZ}_{n}\right\rangle $ is LC-equivalent to an $n$-qubit
star graph state $\left|S_{n}\right\rangle $ (Fig. \ref{fig:star}).
To see this, apply prescription (\ref{eq:gdef}) to generate $\left|S_{n}\right\rangle $
:
\begin{align*}
\left|S_{n}\right\rangle  & =\text{CZ}^{12}\text{CZ}^{13}\ldots\text{CZ}^{1n}\left|+\right\rangle ^{\otimes n}\\
 & =\left|0\right\rangle \left|+\right\rangle ^{\otimes n-1}+\left|1\right\rangle \left|-\right\rangle ^{\otimes n-1},
\end{align*}

where $\left|-\right\rangle \equiv\frac{1}{\sqrt{2}}\left(\left|0\right\rangle -\left|1\right\rangle \right)$.
Then, applying the gate
\[
\mathcal{H}\equiv I^{1}\otimes H^{2}\otimes\cdots\otimes H^{n}
\]
to $\left|S_{n}\right\rangle $, where the $H^{i}=\frac{1}{\sqrt{2}}\begin{bmatrix}1 & -1\\
1 & 1
\end{bmatrix}$ are Hadamard gates, transforms the state to $\left|\text{GHZ}_{n}\right\rangle $.
Since Hadamard gates are local Clifford unitaries, so is $\mathcal{H}$.
But star graphs have no cycles of length three or four, implying,
by Result \ref{enu:cycle}, that they satisfy $\text{LU}\Leftrightarrow\text{LC}$.
Finally, Fact \ref{fact:LCclass} implies that $\left|\text{GHZ}_{n}\right\rangle $
also satisfies $\text{LU}\Leftrightarrow\text{LC}$.
\end{proof}
The next result follows immediately:
\begin{cor}
Complete graph states satisfy $\text{LU}\Leftrightarrow\text{LC}$.
\end{cor}
\begin{proof}
Applying a local complementation to the central qubit of the star
graph, $S_{n}$, produces a complete graph, $C_{n}$ (Fig. \ref{fig:c10}).
By the LC Rule, this implies that $\left|S_{n}\right\rangle $ and
$\left|C_{n}\right\rangle $ are LC-equivalent, and the result follows
by Fact \ref{fact:LCclass}.
\end{proof}
Somewhat less obviously, it is also true that $\text{LU}\Leftrightarrow\text{LC}$
holds for complete\emph{ bipartite }graph state, or bicliques, as
defined and illustrated in Fig. \ref{fig:biclique}:
\begin{claim}
\emph{Biclique states satisfy $\text{LU}\Leftrightarrow\text{LC}$.
\label{claim1}}
\end{claim}
\begin{proof}
We may not use Result \ref{enu:cycle} right away: although no bipartite
graph can have a cycle of length 3, bicliques have numerous length-4
cycles. But Fig. \ref{fig:bicproc} shows how three applications of
the LC Rule transform the biclique to a ``binary star'' state (Fig.
\ref{fig:bistar}). Since the binary star has no cycles of length
3 or 4, Fact \ref{fact:LCclass} and Result \ref{enu:cycle} now give
the result. 
\end{proof}
Actually, the procedure in the proof to Claim \ref{claim1} is easily
generalized to account for asymmetrical bicliques with a different
number of vertices in the left and right sets. See Fig. \ref{fig:genbi}.
In what follows, we will use ``biclique'' to mean ``generalized
biclique,'' as this is the convention in the literature.

In proving the previous claim, we presented an \emph{LC orbit: }the
set of states LC-equivalent to a biclique state. Actually, considering
(i) the symmetry of the problem; (ii) the fact a local complementation
at a leaf vertex has no effect on the graph; and (iii) the fact that
the LC operation is a self-inverse, we have provided the complete
LC orbit up to the permutation of vertices. In general, although it
is easy to identify whether two given graph states are LC-equivalent
\citep{Bouchet1991,VanNest2004}, it is not known whether the LC orbit
of an arbitrary graph state can be generated efficiently \citep{Hein2006a}.

To move from the bare complete graph and biclique states to the repeater
graph states, it is tempting to use Results \ref{enu:msc} or \ref{enu:mscleaf}.
However, none of the graphs we have discussed satisfy the MSC:
\begin{claim}
\emph{Complete graph, biclique, and repeater graph states do not satisfy
the Minimal Support Condition.}
\end{claim}
\begin{proof}
It was shown in \citep{Zeng2007a} that distance-2 stabilizer states,
to which graph states with leaves belong, do not satisfy the MSC.
Therefore repeater graph states, along with star (binary star) graph
states do not satisfy the MSC. But, from \citep{VanDenNest2005,Nest2008},
we know that the function $A_{\omega}\left(\left|\psi\right\rangle \right)$
that gives the number of elements in the stabilizer of $\left|\psi\right\rangle $
with support $\omega$ is invariant under local unitaries. Therefore
the distance of a graph state is LU-invariant, and so all states LC-equivalent
to the star (binary star) graph state, namely the complete graph (biclique)
state, are distance-two. Thus the complete graph (biclique) states
also do not satisfy the MSC. 

Alternatively, we know that complete graph states are locally equivalent
to $n$-GHZ states, which do not satisfy the MSC, as shown directly
in \citep{Nest2008}. 
\end{proof}
The next approach is to append a leaf to each vertex in the complete
graph and biclique state, and then run the local complementation procedures
in Fig. \ref{fig:GHZs} and \ref{fig:bicproc} backwards. This gives
us the following result:
\begin{claim}
\emph{Complete-graph-based repeater graph states with one leaf missing
and biclique-based repeater graph states with two leaves missing satisfy
$\text{LU}\Leftrightarrow\text{LC}$.}
\end{claim}
\begin{proof}
Follow the LC procedure in Fig. \ref{fig:GHZs} and \ref{fig:bicproc}
but with the states in Fig. \ref{fig:Imperfect-repeater-graph}.
\end{proof}
It should be noted that running the LC procedure on the perfect repeater
graph state will not work: the leaf qubit will end up adjacent to
all the neighbours of its parent, introducing cycles of lengths that
preclude the use of Result \ref{enu:cycle}. A different strategy
is thus needed to verify if $\text{LU}\Leftrightarrow\text{LC}$ for
ideal repeater graph states. For now, this question remains open.

\begin{figure}
\begin{centering}
\subfloat[A star graph state, $\left|S^{n}\right\rangle $, where $n=10$. A
central qubit is connected to $n-1$ leaf qubits. \label{fig:star}]{\begin{centering}
\includegraphics[scale=0.35]{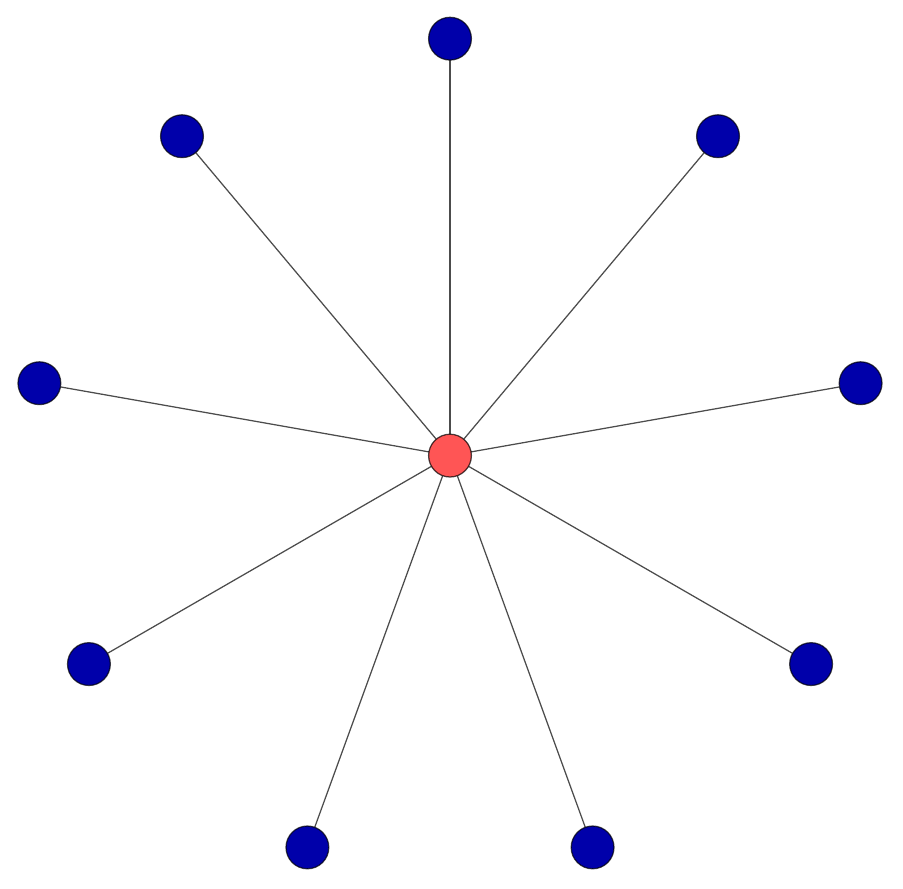}
\par\end{centering}
}$\overset{\text{LC}}{\longleftrightarrow}$\subfloat[A complete graph state, $\left|C^{n}\right\rangle $, where $n=10$.
Every qubit is connected to every other qubit. \label{fig:c10}]{\begin{centering}
\includegraphics[scale=0.35]{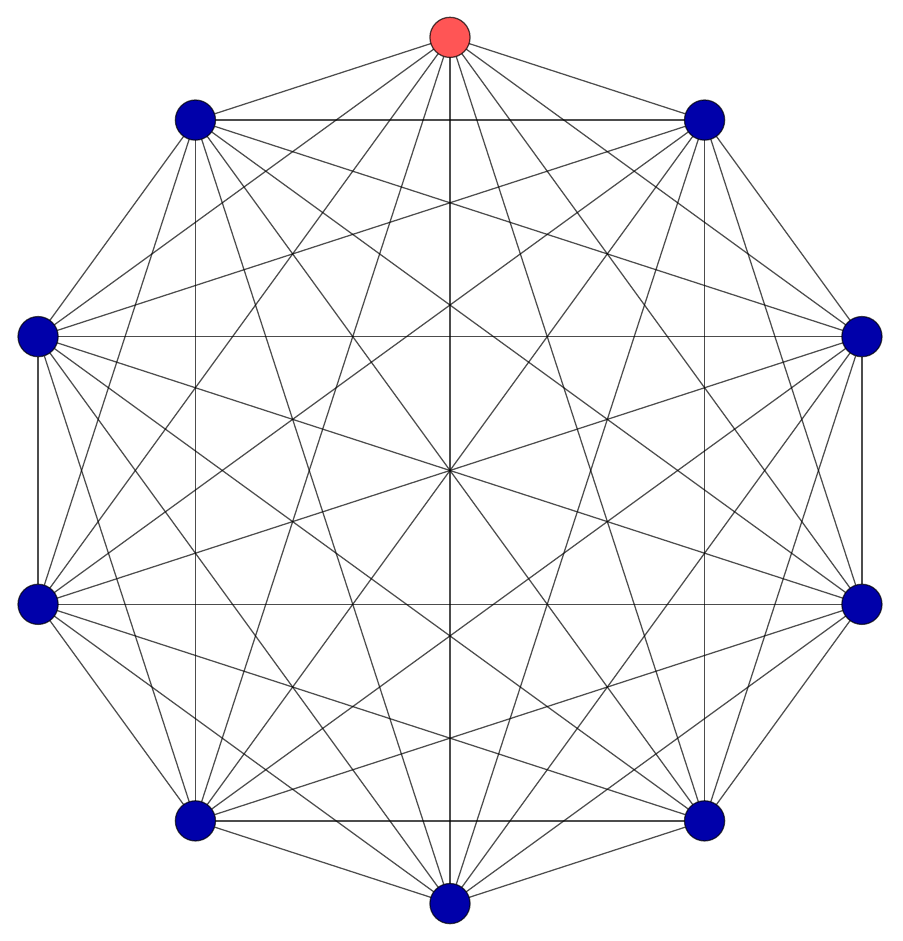}
\par\end{centering}
}
\par\end{centering}
\caption{Stars and complete graphs are locally equivalent. Local complementation
at the red (lighter) vertex on each graph produces the other graph.
\label{fig:GHZs}}
\end{figure}
\begin{figure}
\begin{centering}
\subfloat[A ``binary star'' state, $\left|\Sigma^{2m}\right\rangle $, where
$m=5$. There are two central qubits, each connected to $m-1$ leaves.
\label{fig:bistar}]{\begin{centering}
\includegraphics[scale=0.35]{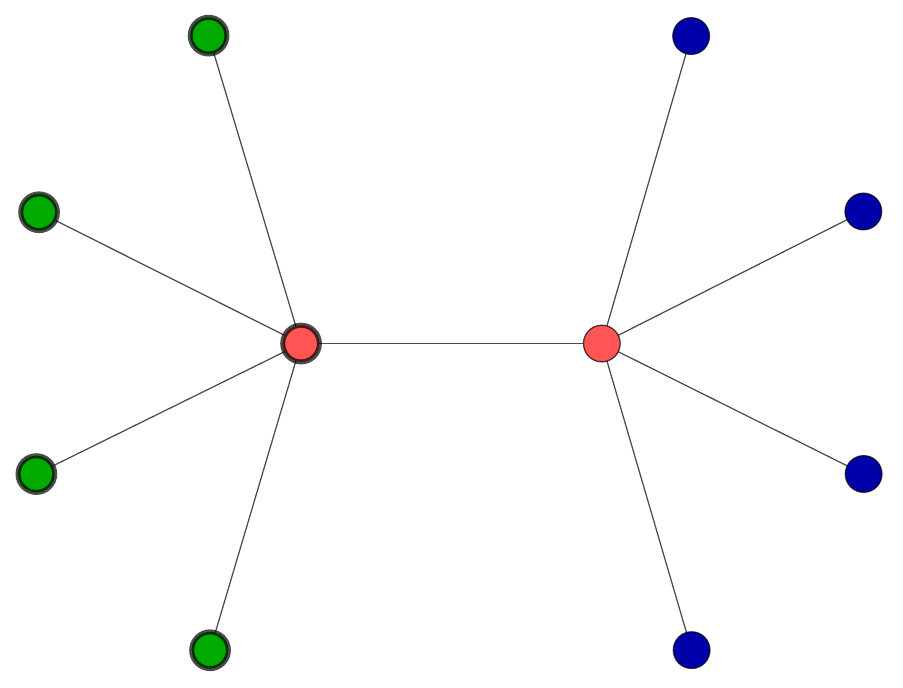}
\par\end{centering}
}$\overset{\text{LC}}{\longleftrightarrow}$\subfloat[A biclique (or complete bipartite graph) state, $\left|B^{2m}\right\rangle $,
where $m=5$. \emph{Complete} means each of the $m$ qubits on the
left is connected to each of the $m$ qubits on the right; \emph{bipartite}
means the qubits within each set are disconnected. \label{fig:biclique}]{\begin{centering}
\includegraphics[scale=0.3]{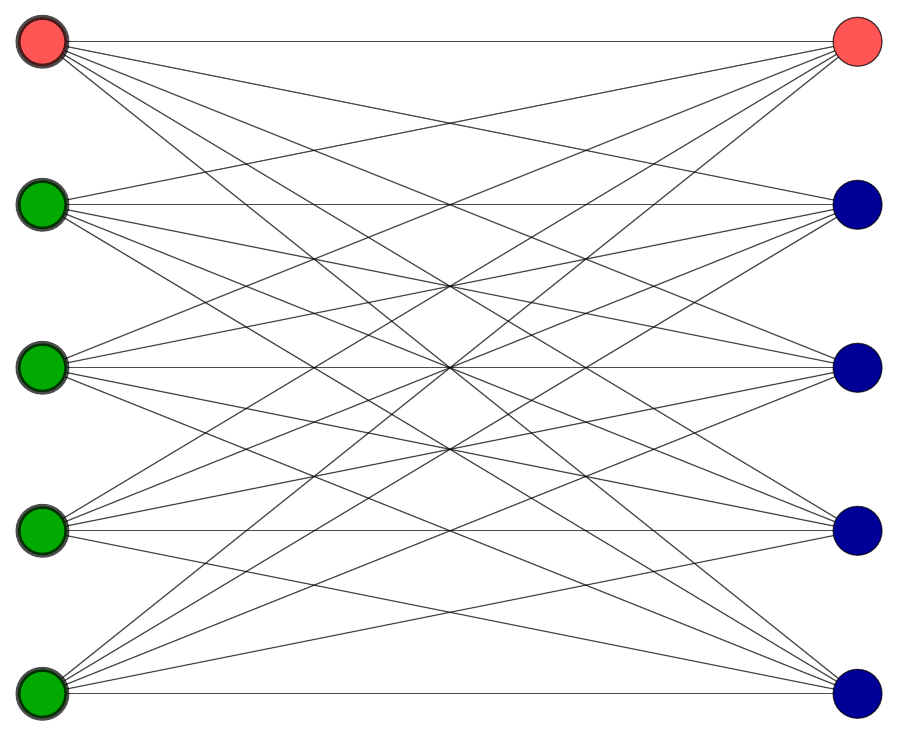}
\par\end{centering}
}
\par\end{centering}
\centering{}\caption{Binary stars and bicliques are locally equivalent. Local complementation
at the central vertices. More detailed procedure in Fig. \ref{fig:bicproc}.}
\end{figure}
\begin{figure}
\begin{centering}
\subfloat[]{\includegraphics[scale=0.35]{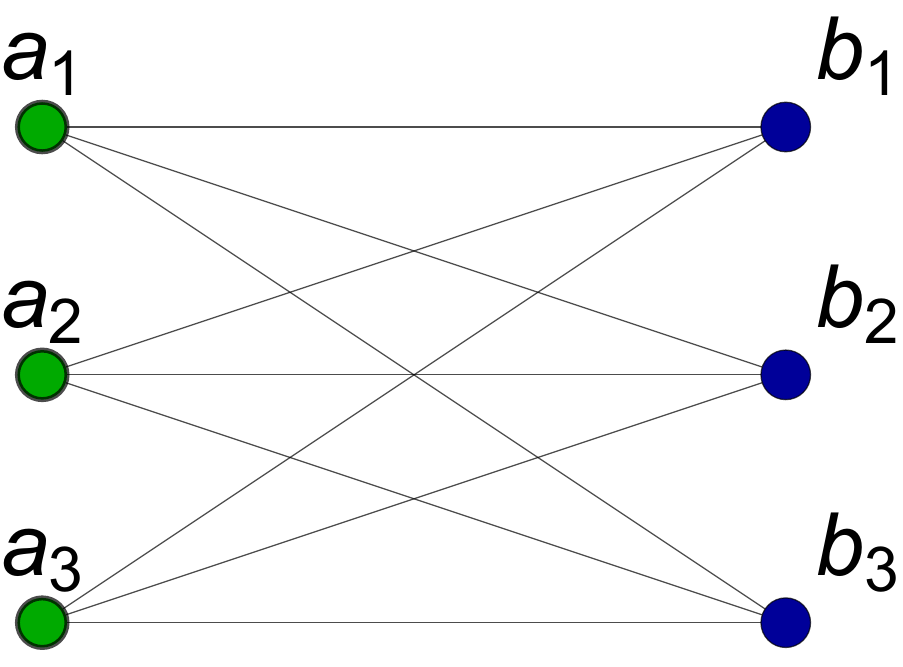}

}$\qquad$\subfloat[]{\includegraphics[scale=0.35]{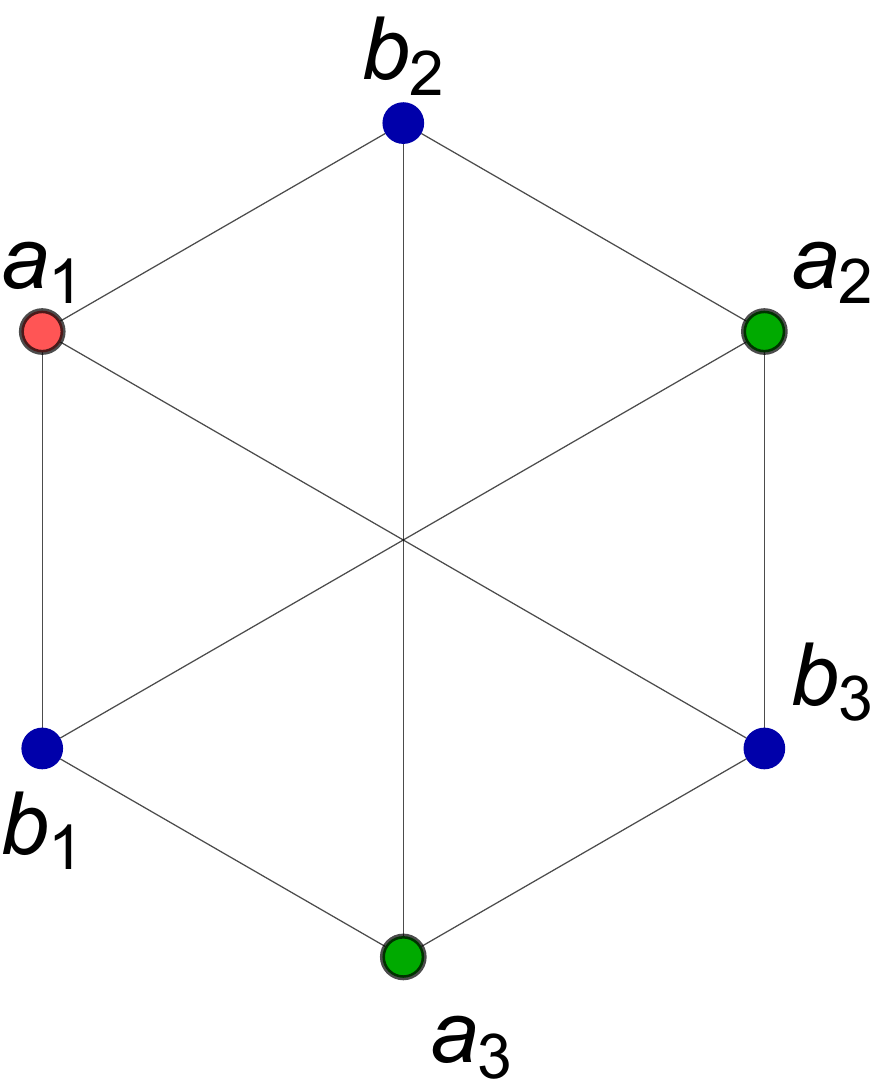}}
\par\end{centering}
\begin{centering}
\subfloat[]{\includegraphics[scale=0.35]{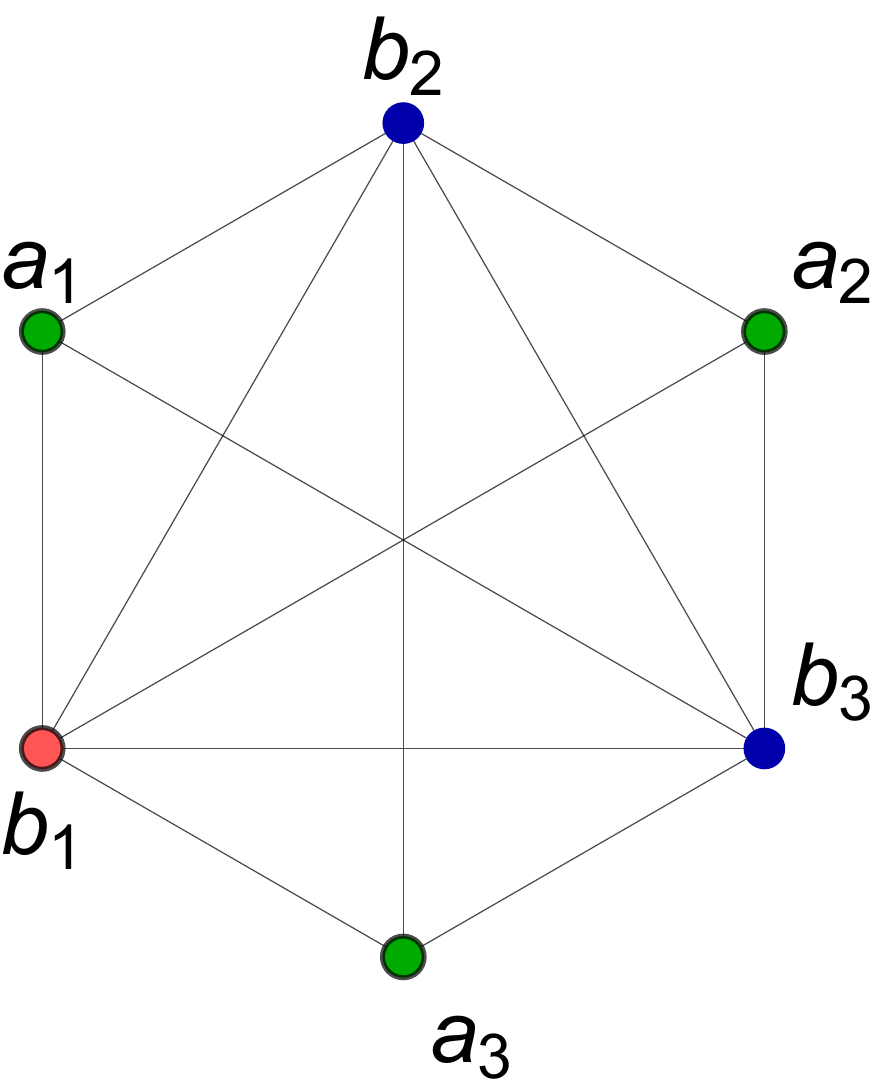}}$\qquad$\subfloat[]{\includegraphics[scale=0.35]{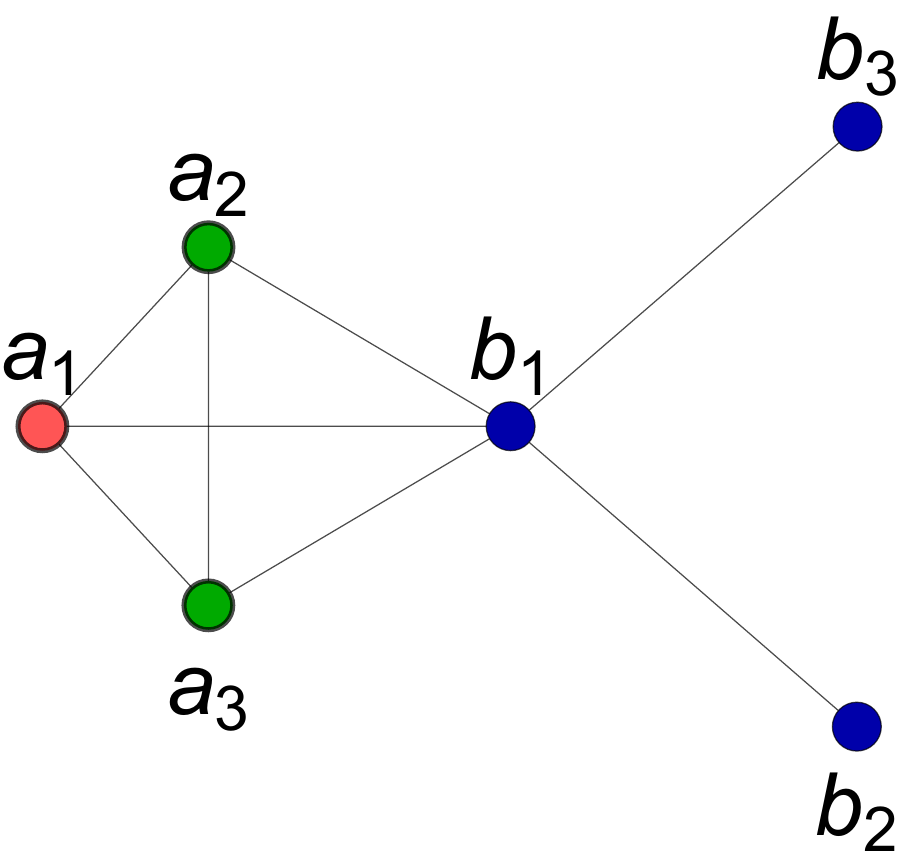}}
\par\end{centering}
\begin{centering}
\subfloat[]{\includegraphics[scale=0.35]{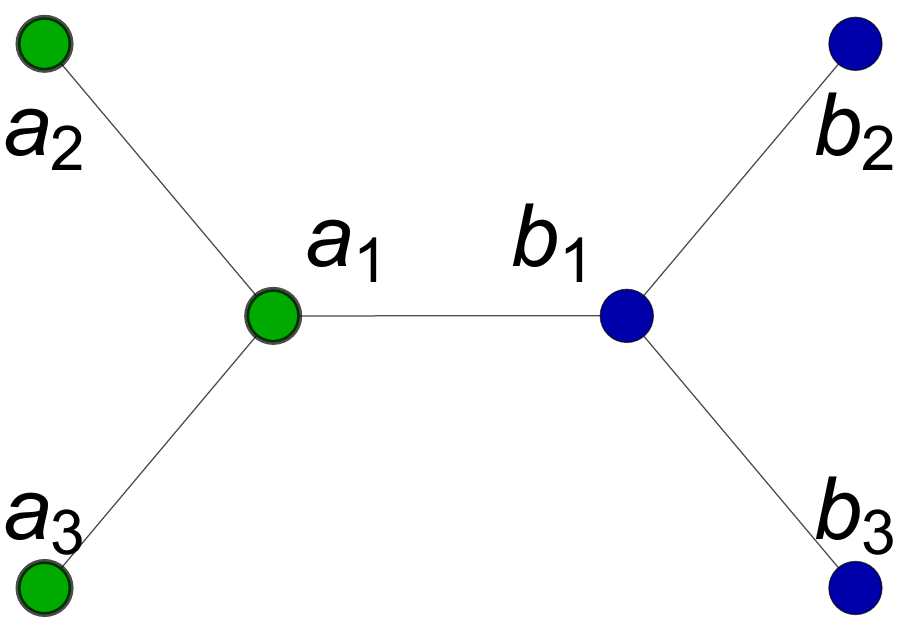}}
\par\end{centering}
\caption{Procedure for transforming biclique state, $\left|B^{2m}\right\rangle $,
into a ``binary star'' state, $\left|\Sigma^{2m}\right\rangle $,
through successive applications of the LC Rule. Here $m=3$ but the
procedure is general. The nodes on the left half are labelled $a_{i}$
while the nodes on the right half are labelled $b_{i}$. First, unravel
the state in (a) to get (b). Then, locally complement the graph in
(b) at $a_{1}$ to obtain (c), in which the $b_{i}$ form a complete
graph. Next, locally complement the graph in (c) at $b_{1}$ to get
a complete graph among $b_{1}$ an all the $a_{i}$. Here $b_{2},\ldots,b_{m}$
are connected to $b_{1}$ but not to each other. Finally, locally
complement the graph in (d) at $a_{1}$ to obtain (e). \label{fig:bicproc}}
\end{figure}
\begin{figure}
\begin{centering}
\subfloat[A generalized binary star state, $\left|\bar{\Sigma}{}^{m,n}\right\rangle $,
where $m=3$ and $n=5$. This is the state in \ref{fig:bistar} but
each of the two central qubits is connected to a different number
of leaves.]{\begin{centering}
\includegraphics[scale=0.35]{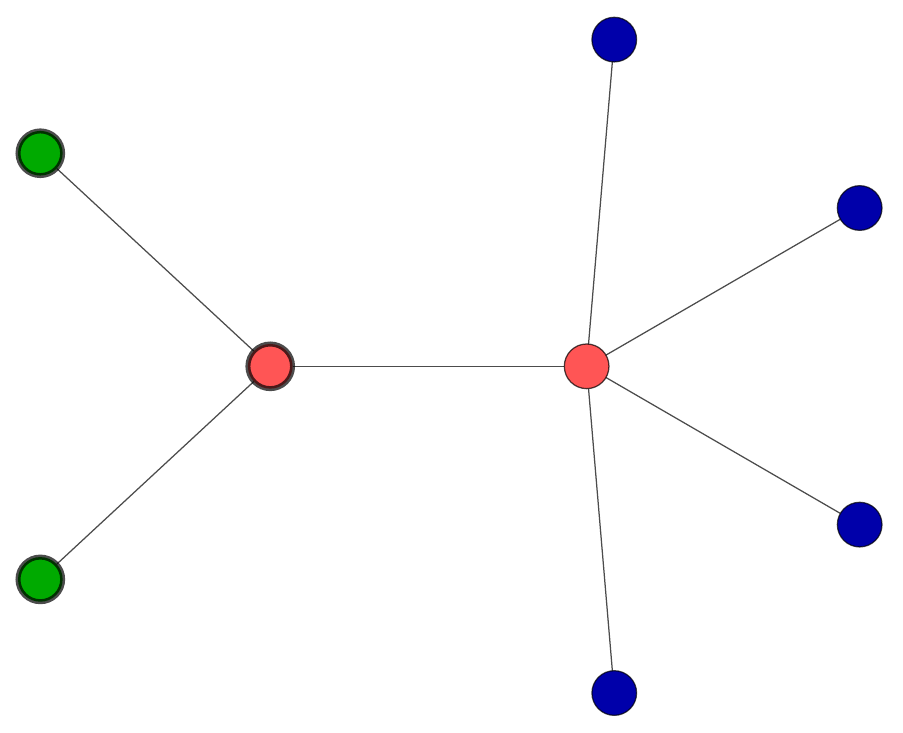}
\par\end{centering}
}$\overset{\text{LC}}{\longleftrightarrow}$\subfloat[A generalized biclique state, $\left|\bar{B}^{m,n}\right\rangle $,
where $m=3$ and $n=5.$ This is the state in \ref{fig:biclique}
but each of the halves has a different number of qubits.]{\begin{centering}
\includegraphics[scale=0.35]{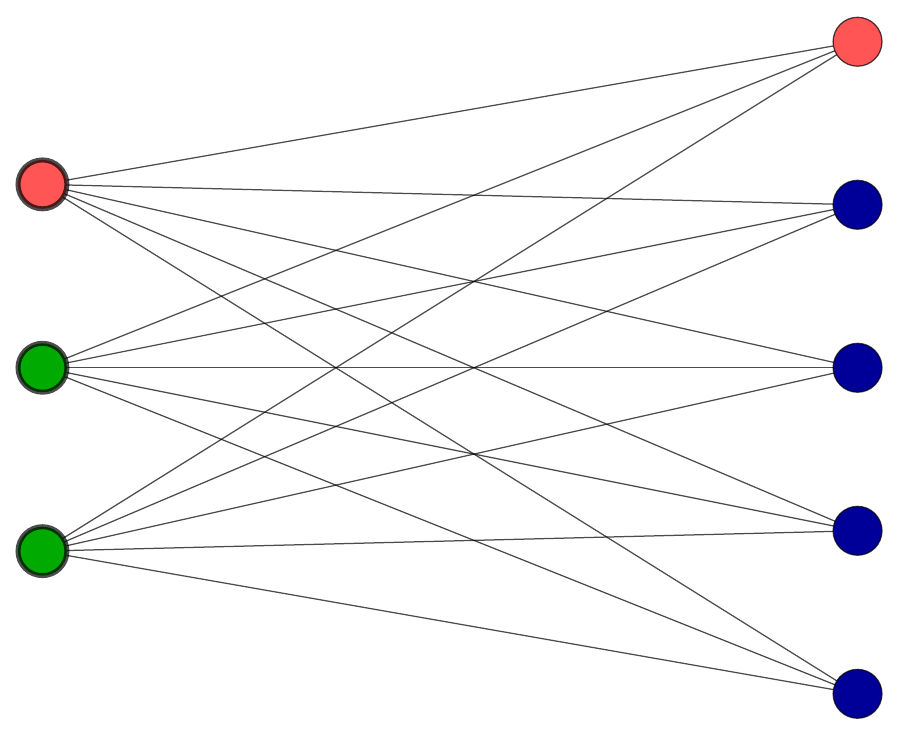}
\par\end{centering}
}
\par\end{centering}
\centering{}\caption{Generalized binary stars and bicliques are locally equivalent. The
steps in Fig. \ref{fig:bicproc} hold but the $\left(m+n\right)$-gons
in steps (b) and (c) will have some missing sides\label{fig:genbi}.}
\end{figure}
\begin{figure}
\begin{centering}
\subfloat[]{\includegraphics[scale=0.35]{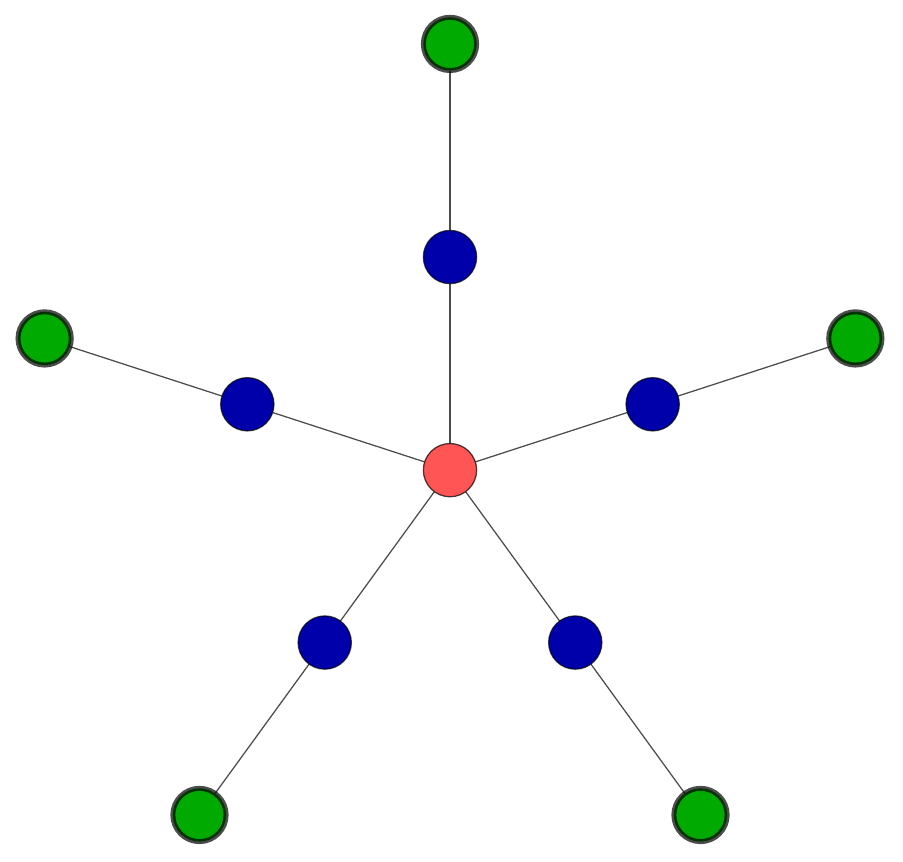}}$\overset{\text{LC}}{\longleftrightarrow}$\subfloat[]{\includegraphics[scale=0.35]{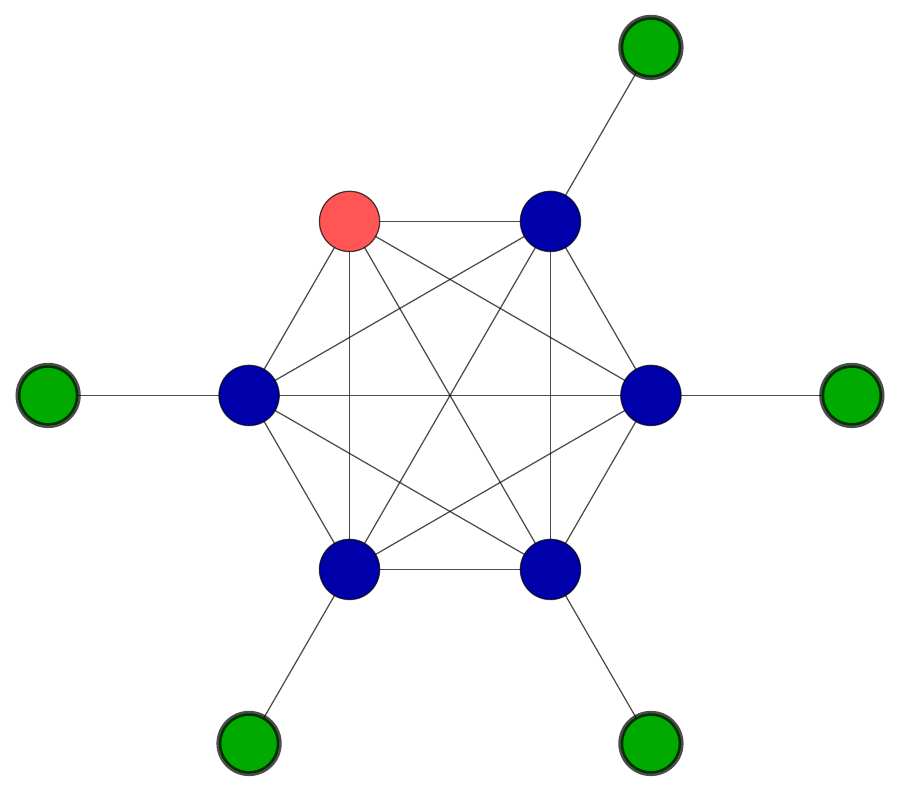}}
\par\end{centering}
\begin{centering}
\subfloat[]{\includegraphics[scale=0.35]{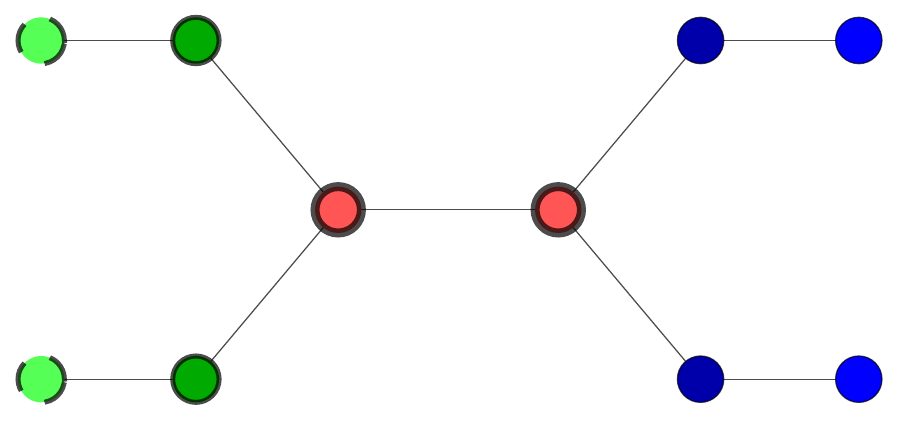}}$\overset{\text{LC}}{\longleftrightarrow}$\subfloat[]{\includegraphics[scale=0.35]{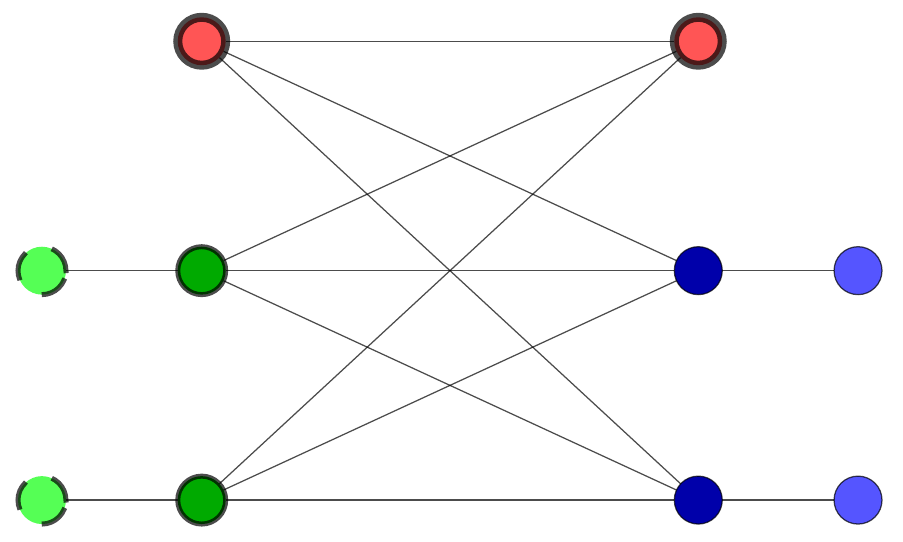}}
\par\end{centering}
\centering{}\caption{Imperfect repeater graph states (b, d) and their LC-equivalent graph
states (a, c). Local complementations at the red (thicker edge) vertices.
The more detailed transformation from c to d follows the steps in
Fig. \ref{fig:bicproc}. \label{fig:Imperfect-repeater-graph}}
\end{figure}

\subsection{More on bicliques \label{subsec:Bicliques}}

Although complete graphs were introduced in \citep{Azuma2015} as
the core of the all-photonic quantum repeater, bicliques may freely
be used, as the protocol does not require connections among the left
and among the right halves of the states. This was noticed independently
in \citep{Russo2018}. As a result, a biclique of $n$ vertices does
away with $n\left(n-2\right)/4$ edges from the complete graph without
affecting the functionality of the repeater. A reduction in edges
does not imply a reduction in complexity, however; after all, the
complete graph, with $n\left(n+1\right)/2$ edges, has only one other
state in its LC orbit (up to a permutation of vertices), a simple
star graph with $n-1$ edges. But an $n$-vertex biclique is locally
equivalent a ``binary star,'' which we named as such because it
resembles the fusion of two half-sized stars of $n/2$ vertices. These
states too have $n-1$ edges, offering a comparable complexity for
possible advantages.

For one, there might be an experimental graph state generation protocol
for which fewer resources are required to implement the fusion of
two stars rather than the creation of one large star. Such a fusion
operation exists for photons, for example \citep{Browne2005}, and
photonic repeater graph states can be created through linear optics
\citep{Azuma2015,Varnava2006,Varnava2007,Varnava2008} or through
deterministic emitters \citep{Lindner2009,Economou2010,Buterakosa,Schwartz2016}.
The discussion in the previous paragraph is relevant to protocols
involving multiple solid-state emitters, as the number of edges in
a graph state is related to the number of CZ gates required to couple
the quantum dots \citep{Buterakosa,Russo2018}. 

Another small benefit of the biclique is in the degree of its entanglement.
Two functions quantifying entanglement for graph states are \citep{Briegel2001,Eisert2001,Hein2006a}:
\begin{itemize}
\item \emph{Pauli persistency}: \emph{PP$\left(\left|G\right\rangle \right)$},
the minimal number of local Pauli measurements required to completely
disentangle $\left|G\right\rangle $.
\item \emph{Schmidt measure}: $E_{S}\left(\left|G\right\rangle \right)=\log_{2}$$\left(\min_{\left|G\right\rangle }k\right)$,
where the minimum is taken over all decompositions
\[
\left|G\right\rangle =\sum_{i=1}^{k}\alpha_{i}\left|g_{i}^{1}\right\rangle \otimes\ldots\otimes\left|g_{i}^{n}\right\rangle 
\]
for $\alpha_{i}\in\mathbb{C}$, $\left|g_{i}\right\rangle \in\mathbb{C}^{2}$,
and $V_{G}=\left\{ 1,\ldots,n\right\} $.
\end{itemize}
Pauli Persistency is a metric by which cluster states – graph states
corresponding to lattices in some dimension – were introduced as a
powerful resource for quantum computation \citep{Briegel2001}. It
provides an upper bound for the Schmidt measure – a proper entanglement
monotone – and the measures coincide for trees, that is, connected
graphs without cycles \citep{Eisert2001,Hein2006a}. Furthermore,
we need only consider Pauli $Z$ measurements for calculating the
Pauli Persistency of tree graph states \citep{Hein2006a}. Because
it takes at least one $Z$ measurement on the central qubit of the
star graph state, which is a kind of tree graph, to disentangle it,
and because the Pauli Persistency and the Schmidt measure are invariant
under local Clifford unitaries, we conclude that $PP=E_{S}=1$ for
star graph and complete graph states. Similarly, since at least two
$Z$ measurements on the central qubits of the binary star graph state
are required disentangle it, $PP=E_{S}=2$ for binary star graph and
biclique states. Although these numbers are small – $n$-qubit cluster
states have a persistency (and Schmidt measure) of $\left\lfloor \frac{n}{2}\right\rfloor $,
for comparison – bicliques offer a modest improvement in the robustness
of their entanglement over complete graphs. If one star with $n/2$
qubits (i.e. an $\frac{n}{2}$-GHZ state) is decoupled from the graph,
another remains.

Complete bipartite graphs also feature in Rudolph's treatise on silicon-photonic
quantum computing \citep{Rudolph2016b}. Rudolph does not suggest
a biclique as a replacement for the complete graph to underlie the
all-photonic quantum repeater, but he describes a \emph{crazy graph
}(Fig. \ref{fig:crazygraph3}) that can be used for a loss-tolerant
encoding scheme. In fact, our result on bipartite graph states applies
to crazy graph states of small sizes:
\begin{claim}
$\text{LU}\Leftrightarrow\text{LC}$ holds for crazy graph states
with two or three columns.
\end{claim}
\begin{proof}
A crazy graph with two columns is simply a symmetrical bipartite graph.
For a crazy graph with three columns, call the vertex sets $L$, $M$
and $R$ for left, middle and right. Suppose there are $m$ vertices
in each column. Notice first that the vertices within $L$, $M$ and
$R$ are disconnected. Notice further that no vertex in $L$ is adjacent
to a vertex in $R$, whereas all the vertices in $M$ are adjacent
to all the vertices in $L\cup R$. Thus we may rewrite this crazy
graph as an asymmetrical biclique state, $\left|\bar{B}^{2m,m}\right\rangle $
(Fig. \ref{fig:crazygraph3bip}), and then use Claim \ref{claim1}.
\end{proof}
If we add more columns to the crazy graph, it still has a bipartite
embedding (position all the odd columns as left vertices, and all
the even columns as right vertices), but it is no longer a biclique,
and so we may not use our results from above. Similarly, the leaves
in repeater graphs ensure that the best bipartition is not complete.

We note next that bipartite graphs are the \emph{two-colourable }graphs,
meaning their vertices can be painted in two different colours such
that no two vertices of the same colour are adjacent. These states
are locally equivalent to \emph{Calderbank-Shor-Steane }(CSS) states
that arise in quantum error correction \citep{Calderbank1996,Chen2004,Steane1996}
and feature in schemes for multipartite cryptography \citep{Chen2004}
and entanglement purification \citep{Aschauer2005}. In fact, Sarvepalli
and Raussendorf in \citep{Sarvepalli2010a} show that $\text{LU}\Leftrightarrow\text{LC}$
holds for a certain class of CSS states. However, these CSS states
are associated to classical error-correcting codes of distance (and
dual distance) of 3 or greater, to which bicliques do not belong,
as we now show:
\begin{claim}
The biclique state corresponds to a CSS code of either distance or
dual distance $2$. Therefore we may not rely on Corollary 3 in \citep{Sarvepalli2010a}
to check if biclique states satisfy $\text{LU}\Leftrightarrow\text{LC}$.
\label{csscodeclaim}
\end{claim}
\begin{proof}
Consider a biclique of $m$ qubits on the left and $n$ qubits on
the right. We can write the\emph{ }check matrix of the stabilizer
of this state as follows:
\[
\left[\begin{array}{cc|cc}
I_{m} & 0 & 0 & A\\
0 & I_{n} & A^{T} & 0
\end{array}\right],
\]
where $I_{\alpha}$ is the $\alpha\times\alpha$ identity matrix,
$\alpha\in\left\{ m,n\right\} $, and $A$ is an $m\times n$ matrix
consisting only of 1s. Here we use the convention that the $X$ $\left(Z\right)$
operators correspond to the left (right) half of the matrix. Following
the argument in Claim 2 of \citep{Chen2004} backwards, we may apply
the Hadamard matrix to the final $n$ vertices and rearrange the columns
to get the new check matrix
\[
\left[\begin{array}{cc|cc}
I_{m} & A & 0 & 0\\
0 & 0 & A^{T} & I_{n}
\end{array}\right].
\]

But this is a check matrix of a CSS state, since it corresponds to
the stabilizer whose elements are products of $X$ and the the identity
for the first $m$ qubits, and $Z$ and the identity for the following
$n$ qubits \citep{Gottesman1997}. From here, we can read off the
parity check matrix for the code, $H\left(C\right)=\left[I_{m}\vert A\right]$.
But, from \citep{Nielsen2010}, we know that the distance of the code
is the minimum number of linearly dependent column vectors in the
parity check matrix. Since $I_{m}$ always has $m$ linearly independent
vectors, and $A$ consists of $n$ columns of 1s, we conclude that
$d=2$ for $n\geq2$. If $n=1$, we can look at the parity matrix
$H\left(C^{\perp}\right)=\left[A^{T}\vert I_{n}\right]$ of the dual
code. This matrix has $m+1$ columns of 1s, meaning the distance of
$C^{\perp}$ – the dual distance – is 2.
\end{proof}
\begin{figure}
\begin{centering}
\subfloat[Crazy graph state with three columns and $m=3$ qubits per column.
\label{fig:crazygraph3}]{\includegraphics[scale=0.35]{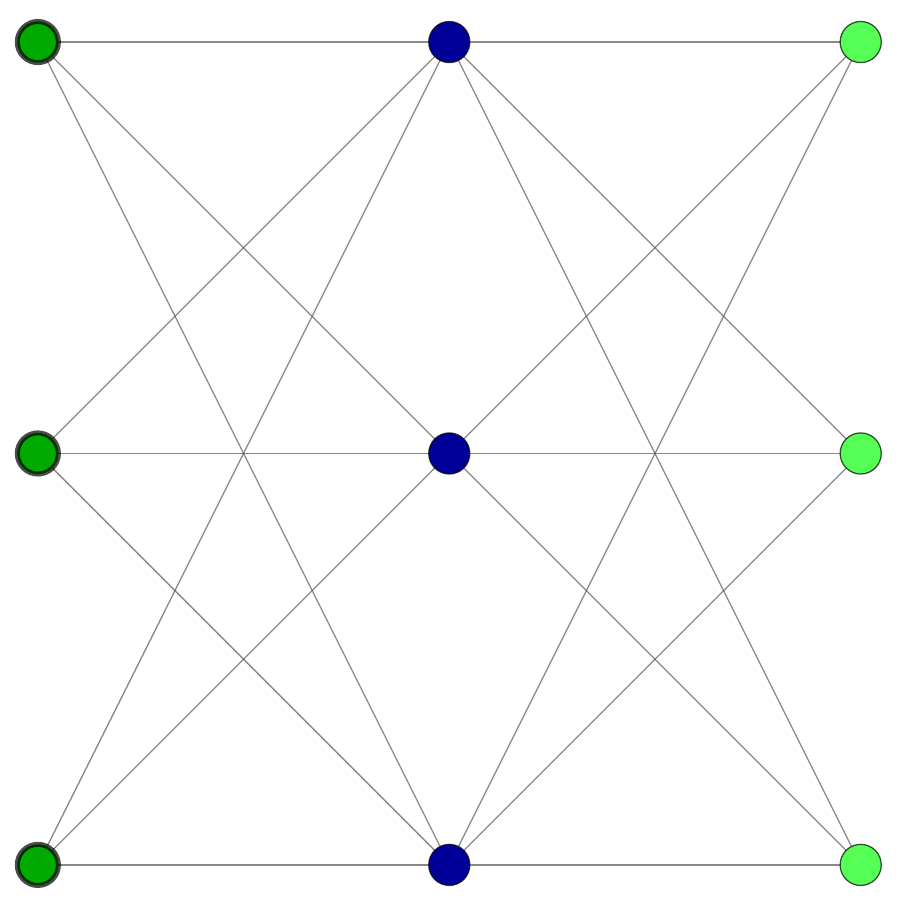}}$\quad\equiv\quad$\subfloat[Crazy graph state from (a) with a bipartite embedding. This is an
asymmetrical biclique state, $\left|\bar{B}^{2m,m}\right\rangle $,
with $m=3$. \label{fig:crazygraph3bip}]{\includegraphics[scale=0.4]{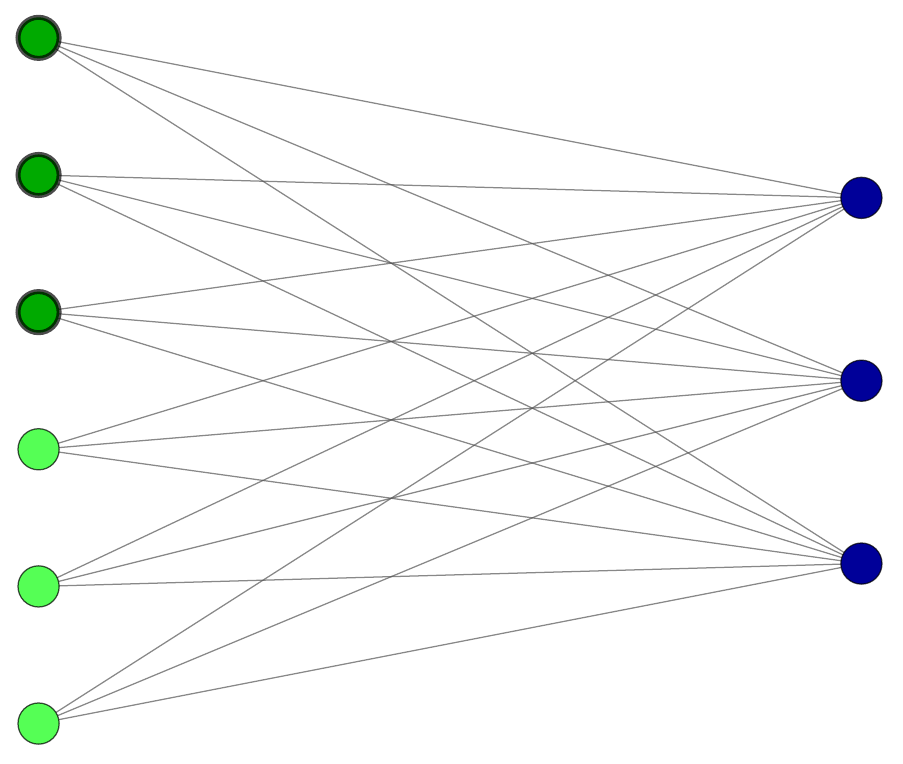}}
\par\end{centering}
\caption{Crazy graph state from \citep{Rudolph2016b}, used for loss-tolerant
encoding. \label{fig:Crazy-graph}}
\end{figure}

\subsection{Low-rank stabilizer states \label{subsec:counterex}}

We conclude Section \ref{sec:Results} by clarifying what is meant
by Result \ref{enu:rankstab} in Section \ref{subsec:Review-of-results}
and whether it applies to our discussion. In \citep{Ji2007}, Ji et
al. were the first to provide a random search algorithm for generating
counterexamples to the LU-LC Conjecture. Their result relies on a
reduction of the LU-LC problem to one involving quadratic forms on
linear subspaces of $\mathbb{F}_{2}^{n}$, the $n$-bit binary field
\citep{Gross2007a}. In \citep{Dehaene,Ji2007}, it was shown that
any stabilizer state can be written as 
\begin{equation}
\left|S\right\rangle =\frac{1}{\sqrt{\left|S\right|}}\sum_{x\in S}i^{l\left(x\right)}\left(-1\right)^{q\left(x\right)}\left|x\right\rangle ,\label{eq:stabin}
\end{equation}
where $l$ is linear, $q$ is quadratic, and $S$ is a linear subspace
of $\mathbb{F}_{2}^{n}$. Ji et al. provide an explicit counterexample
in the case $S$ is a rank $r=6$ subspace of $\mathbb{F}_{2}^{27}$,
and they prove that no stabilizer state corresponding to a smaller-rank
subspace can yield counterexamples. To see the implication for graph
states, consider the following: A stabilizer state can viewed as a
$[\left[n,0,d\right]]$ quantum code \citep{Cross2009b}, and we can
interpret $r$ to be the rank of the $X$ part of the check matrix
of this code \citep{Dehaene}, $T=\left[X\vert Z\right]$. (Actually,
since applying the Hadamard gate to all the qubits swaps the $X$
and the $Z$ parts of $T$, we may also interpret $r$ to be rank
of the $Z$ part of the matrix). As in \citep{Nielsen2010}, we can
perform row reduction on $X$ to obtain
\begin{equation}
T^{\left(1\right)}=\left[\begin{array}{cc|cc}
I_{r} & A & B & C\\
0 & 0 & D & E
\end{array}\right],
\end{equation}
where $A$ is an $r\times\left(n-r\right)$ matrix, followed by row
reduction on $E$ to obtain,
\begin{equation}
T^{\left(2\right)}=\left[\begin{array}{cc|cc}
I_{r} & A & B' & 0\\
0 & 0 & D' & I_{n-r}
\end{array}\right].
\end{equation}

Then, applying Hadamard gates to the final $n-r$ qubits, we have
\begin{equation}
T^{\left(3\right)}=\left[\begin{array}{cc|cc}
I_{r} & 0 & B' & A\\
0 & I_{n-r} & D' & 0
\end{array}\right].
\end{equation}

The stabilizer commutativity condition, $T\Lambda T^{T}=0$, where
$\Lambda=\begin{bmatrix}0 & I_{n}\\
I_{n} & 0
\end{bmatrix}$, implies the right part of $T^{\left(3\right)}$ is symmetric, so
that $D'=A^{T}$ and $B'=B'^{T}$. Finally, we may set the appropriate
qubits of $B'$ to 0, as in \citep{Nest2003}, to obtain $T^{\left(4\right)}=\left[I_{n}\vert\Gamma\right]$,
where
\begin{equation}
\Gamma=\begin{bmatrix}M & A\\
A^{T} & 0
\end{bmatrix}.
\end{equation}
Notice that $\Gamma$ is symmetric and has zeros along its diagonal,
meaning it is an adjacency matrix of some graph, $G_{S}$. Furthermore,
from $T\to T^{\left(4\right)}$ all we have done is apply local Clifford
operations and a basis change. This means \citep{Nest2003} that the
graph state described by the stabilizer matrix $T^{\left(4\right)}$
is LC-equivalent to $T$, and hence the corresponding graph state
$\left|G_{S}\right\rangle $ is LC-equivalent to $\left|S\right\rangle $.
When $M\neq0$, $\Gamma$ corresponds to a graph with some edges among
the qubits $R=\left\{ 1,\ldots,r\right\} $ but none among $L=\left\{ r+1,\ldots,n\right\} $.
When $M=0$, $\Gamma$ represents a bipartite graph with a bipartition
$\left(L,R\right)$. For a fixed stabilizer subspace, $S$, the topology
of the graph – in other words, the exact form of $A$ and $M$ – depends
on the functions $l\left(x\right)$ and $q\left(x\right)$ in Eq.
\ref{eq:stabin}. We can make some general statements, however.

First, let us disambiguate some confusion that exists in physics literature
(e.g. \citep{Hein2006a,Cabello2009,Cosentino2009}) regarding the
term ``binary rank.'' For this, we enumerate three kinds of ranks
of binary matrices that are used in computer science and mathematics
literature \citep{Gregory1991,Watson2015}. Assume $A$ is an arbitrary
$r\times\left(n-r\right)$ matrix with entries in $\left\{ 0,1\right\} $,
and that $A$ is associated with the upper-right block of the adjacency
matrix of a bipartite graph, $G_{A}$, with bipartition $\left(R,L\right)$:
i.e., its rows represent the vertices in $R$, and the columns represent
the ones in $L$. Without loss of generality, suppose $r\leq n-r$.
Then, consider the following definitions:
\begin{itemize}
\item \emph{Binary rank} ($\text{rank}_{\mathbb{C}}$): Normal arithmetic
(row reduction), done in quantum information over $\mathbb{C}$. Equal
to the \emph{biclique partition number}, $bp$: the least number of
bicliques needed to partition every edge of $G_{A}$. This means every
edge of $G_{A}$ is in \emph{exactly one} biclique of the partition
\citep{Gregory1991}. 
\item \emph{Boolean rank} ($\text{rank}_{B}$): Boolean algebra, with addition
and multiplication defined through
\[
\begin{array}{c|cc}
\oplus & 0 & 1\\
\hline 0 & 0 & 1\\
1 & 1 & 1
\end{array}\qquad\begin{array}{c|cc}
\odot & 0 & 1\\
\hline 0 & 0 & 0\\
1 & 0 & 1
\end{array}
\]

Equal to the \emph{biclique cover number}, \emph{bc: }the least number
of bicliques needed to cover every edge of $G_{A}$. This means every
edge of $G_{A}$ is in \emph{at least one} biclique of the cover \citep{Gregory1991}.
\item \emph{XOR rank} ($\text{rank}_{X}$):\emph{ }Algebra over $\mathbb{F}_{2}$
(i.e. mod-2 arithmetic), with addition and multiplication defined
through \citep{Watson2015}:
\[
\begin{array}{c|cc}
\oplus & 0 & 1\\
\hline 0 & 0 & 1\\
1 & 1 & 0
\end{array}\qquad\begin{array}{c|cc}
\odot & 0 & 1\\
\hline 0 & 0 & 0\\
1 & 0 & 1
\end{array}.
\]
\end{itemize}
Generally one really means the XOR rank, rather than the binary rank,
when dealing with stabilizer operations, since arithmetic is performed
modulo 2. Because a set of vectors linearly independent over $\mathbb{C}$
might not be over $\mathbb{F}_{2}$, but the opposite is impossible,
we have that
\begin{equation}
\text{rank}_{\text{X}}\left(A\right)\leq\text{rank}_{\mathbb{C}}\left(A\right)=bp\left(A\right)\leq r.
\end{equation}

Hence the XOR rank of $A$ is bounded above by the biclique partition
number. For bicliques, we saw in the proof to Claim \ref{csscodeclaim}
that $\text{rank}_{X}\left(A\right)$ =1, and $bp=1$ by definition.
This inequality, of course, does not preclude bicliques as counterexamples
to LU-LC with the approach in \citep{Ji2007}. If we would like to
say a little more about the relationship between matrix $A$ and the
decomposition (\ref{eq:stabin}), consider the following: From \citep{Hein2006a},
we know $\text{rank}_{X}$ is equal to the \emph{Schmidt rank}, $k_{S}^{R}$,
of the graph state $\left|G_{A}\right\rangle $ with respect to the
bipartition $\left(R,L\right)$. According to \citep{Cosentino2009},
we may relate the Schmidt rank to $w\left(\left|G\right\rangle \right)$,
the number of minus signs in the computational basis decomposition
of the graph state:
\begin{equation}
k_{S}^{R}=n-\log_{2}\left(2^{n-1}-w\left(\left|G_{A}\right\rangle \right)\right)-1.
\end{equation}
Alternatively, the authors in \citep{Hein2006a} also show that $k_{S}^{R}=\left|R\right|-\text{rank}\left(S_{R}\right)=r-\log_{2}\left|S_{R}\right|$,
where $S_{R}$ denotes the subgroup of the stabilizer generated by
elements whose support is in $R$. Comparing the expressions, we see
that
\begin{align}
\left|S_{R}\right|= & 2^{r}-2^{-\left(n-r-1\right)}w\left(\left|G_{A}\right\rangle \right).
\end{align}
Considering our discussion, then, we have
\begin{equation}
r-\log_{2}\left|S_{R}\right|\leq bp\left(A\right)\leq r,\label{eq:ineq}
\end{equation}

where $R=\left\{ 1,\ldots,r\right\} $. This means that counterexamples
generated by Ji et al.'s procedure, if LC-equivalent to bipartite
graphs, will have a biclique partition number between $r-\log_{2}\left|S_{R}\right|$
and $r$. The range of the biclique partition number can then be determined
by counting the number of minus signs in the standard basis decomposition.
This would require seeing how the negative signs in the coefficients
(\ref{eq:stabin}) change under the local Clifford gate that takes
$\left|S\right\rangle $ to a graph state.

Finally, we see that the above observations are consistent with the
examples given by Ji et al.: there, the authors provide two $n=27$
qubit graph states. One of the graphs is bipartite, with a bipartition
of $r=6$ by $21=n-r$ vertices; the other has an extra edge among
the vertices $1,\ldots,6$. The graphs are proven to be LU equivalent
but not LC equivalent. In \citep{Guhne2017}, it's noted that certain
sequences of local complementations on $\left|G_{S}\right\rangle $
keep it bipartite and keep the bipartitions the same size, indicating
that, perhaps, a necessary condition for counterexamples is $\left|S_{R}\right|=1\implies k_{S}^{R}=r$,
and $w\left(\left|G_{S}\right\rangle \right)=2^{n-1}\left(1-2^{-r}\right)$,
where $r\ge6$ is the rank of the stabilizer support.

\section{Discussion}

Although analytical approaches to questions of local equivalence of
general graph states are difficult, headway can be made by considering
special subsets of graph states. Here we have shown that looking at
local Clifford operations only is sufficient to understand the local
unitary equivalence of complete graph states, generalized biclique
states, imperfect all-photonic quantum repeater graph states and small
crazy graphs, as displayed in Figures \ref{fig:GHZs} to \ref{fig:Crazy-graph}.
Because $\text{LU}\Leftrightarrow\text{SLOCC}$ holds for graph states
generally \citep{Hein2006a}, this has repercussions for an even broader
class of local operations. In turn, for quantum information protocols
involving the aforementioned states, one headache of identifying locally
equivalent states is relieved.

For bicliques, we have also demonstrated the failure of certain approaches
to answer the $\text{LU}\Leftrightarrow\text{LC}$ question: namely,
that these states do not satisfy the Minimal Support Condition in
\citep{Nest2008}, are not $\text{LC}$-equivalent to the class of
CSS codes described in \citep{Sarvepalli2010a}, and do not necessarily
correspond to low-rank stabilizer subspaces in \citep{Ji2007} that
preclude LU-LC counterexamples. We have also provided the LC orbit
of a biclique and shown certain advantages in using this state to
underlie the all-photonic quantum repeater \citep{Azuma2015}. These
boons carry over to other protocols involving a time-reversed entanglement
swapping that is the crux of the all-photonic repeater protocol \citep{Azuma2015}. 

Though we have demonstrated that bicliques satisfy $\text{LU}\Leftrightarrow\text{LC}$,
we saw that the result does not hold for bipartite graphs more generally.
The first explicit counterexample to the LU-LC Conjecture in \citep{Ji2007}
is a 27-qubit state corresponding to a rank-6 stabilizer subspace.
We have striven to clarify the relationship between this rank and
the properties of the corresponding bipartite graph state, but there
is room for further investigation. In \citep{Guhne2017}, Tsimaskuridze
and Gühne highlighted the bipartition of Ji et al.'s state and exploited
it to construct their own counterexamples. They used graphical rules
on non-LC operations aided by so-called \emph{hypergraph }states,
which generalize graph states by allowing an edge to connect to an
arbitrary number of vertices. Whether there are counterexamples below
27 qubits, whether there are more efficient and systematic ways of
generating them, and whether they have convincing physical interpretations
are big open questions in graph state theory that hypergraphs might
help to answer. 

More humbly, it would be interesting to determine whether $\text{LU}\Leftrightarrow\text{LC}$
holds for perfect repeater graph states based either on complete graphs
or a bicliques, and for crazy graph states of an arbitrary number
of columns.
\begin{acknowledgments}
The author would like to thank Hoi-Kwong Lo, Eli Bourassa, Saikat
Guha, Aharon Brodutch, and Koji Azuma for helpful discussions. This
work is financially supported by the Ontario Graduate Scholarship
and the Natural Sciences and Engineering Research Council of Canada
(NSERC).
\end{acknowledgments}

\bibliographystyle{apsrev4-1}
\bibliography{library}

\end{document}